\documentclass[preprint
]{revtex4-1}
\usepackage{graphicx}
\usepackage{amsfonts}
\usepackage{amssymb}
\usepackage{amsthm}
\usepackage{array}
\usepackage{amsmath}
\usepackage{verbatim} 
\usepackage{hyperref}
\usepackage{color}
\usepackage{bbold}
\usepackage{epstopdf}
\usepackage{mathtools}
\usepackage{enumerate}
\usepackage{subcaption}
\usepackage{xcolor}
\usepackage[normalem]{ulem}

\newtheorem{proposition}{Proposition}
\newtheorem*{proposition*}{Proposition}

\newtheorem{theorem}{Theorem}
\newtheorem*{theorem*}{Theorem}
\newtheorem{corollary}{Corollary}[theorem]
\newtheorem*{corollary*}{Corollary}

\newcommand{\ket}[1]{\left\vert#1\right\rangle}
\newcommand{\bra}[1]{\left\langle#1\right\vert}

\newcommand{\abs}[1]{\left|#1\right|}

\def\Tr{\mathrm{Tr}}

\begin{document}
\title{Nontrivial quantum observables can always be optimized via some form of coherence}
\author{Kok Chuan Tan}
\email{bbtankc@gmail.com}
\author{Seongjeon Choi}
\author{Hyunseok Jeong}
\email{jeongh@snu.ac.kr}
\affiliation{Center for Macroscopic Quantum Control \& Institute of Applied Physics, Department of Physics and Astronomy, Seoul National University, Seoul, 151-742, Korea}
\date{\today}

\begin{abstract}
In this paper we consider quantum resources required to maximize the mean values of any nontrivial quantum observable. We show that the task of maximizing the mean value of an observable is equivalent to maximizing some form of coherence, up to the application of an incoherent operation. As such, for any nontrivial observable, there exists a set of preferred basis states where the superposition between such states is always useful for optimizing a quantum observable. The usefulness of such states is expressed in terms of an infinitely large family of valid coherence measures which is then shown to be efficiently computable via a semidefinite program. We also show that these coherence measures respect a hierarchy that gives the robustness of coherence and the $l_1$ norm of coherence additional operational significance in terms of such optimization tasks.
\end{abstract}

\maketitle

\section{Introduction}
Quantum coherence has long been recognized as a fundamental aspect of quantum mechanics. In comparison however, the identification of quantum coherence as a useful and quantifiable resource is a much more recent development. Progress in this area has been greatly accelerated via the introduction of the so-called resource theoretical framework \cite{Aberg2006, Baumgratz2014, Levi2014}. Inspired by the resource theory of entanglement~\cite{Horodecki2001, Horodecki2009}, the notion of what quantum coherence is, as well as how it is to be quantified is now axiomatically defined, thus allowing quantum coherence phenomena to be discussed much more unambiguously. Since this development, many coherence measures have been proposed. Some known measures now include geometric measures\cite{Baumgratz2014}, the robustness of coherence\cite{Napoli2016, Piani2016}, as well as entanglement based measures\cite{Streltsov15}. Coherence measures are have now been studied in relation to a diverse range of quantum effects such as the quantum interference\cite{Wang2017}, exponential speed-up of quantum algorithms\cite{Hillery2016, Matera2016} and quantum metrology\cite{Giorda2016, Tan2018}, nonclassical light~\cite{Tan2017, Kwon2018, Yadin2018}, quantum  macroscopicity~\cite{Yadin2015, Kwon2017} and quantum correlations~\cite{Tan2016, Killoran2016, Regula2017, Wu2017, Ma2016}.  An overview of coherence measures  and their structure may be found in~\cite{Streltsov2017, Streltsov2017rev}.

The computation of such coherence measures usually require full prior knowledge of the input states, which in many cases is difficult to achieve. In contrast, a coherence witness is typically much simpler to implement in the laboratory. The problem of witnessing coherence is equivalent to the problem of constructing some Hermitian observable $W$ which permits positive values $\Tr(\rho W) > 0$ only when $\rho$ is coherent (note that the converse may not be true). In general, for any given coherent state $\rho$, such a witness can always be found \cite{Napoli2016}. 

In this paper, we show that the existence of coherence witnesses may in fact be far more prevalent than one would initially expect. In fact, we demonstrate that every nontrivial Hermitian observable is a witness for at least \textit{some} form of coherence. This suggests that one does not always need additional apparatus in order to detect coherence -- the existence of coherence in many cases may be inferred from existing measurements. We then consider the task of optimizing some objective function $\langle M \rangle $ where $M$ is a quantum observable, and show that the task of optimizing the observable is the same as the task of maximizing the coherence of the input state, up to some incoherent operation. We then show that this leads to an infinitely large class of coherence measures that is computable via a semidefinite program. We also demonstrate that the robustness of coherence and the $l_1$ norm of coherence establishes the quantum limits of such tasks.

\section{Preliminaries} \label{Preliminaries}

We review some elementary concepts concerning coherence measures, quantum channels and semidefinite programs. 

We first briefly describe the formalism of quantum channels, which we take here to mean the set of all Completely Positive, Trace Preserving (CPTP) maps. There are several equivalent characterizations of quantum maps, but for our purposes, we will be concerned with the Kraus~\cite{Kraus} and the Choi-Jamiolkowski representations~\cite{Choi, Jamiolkowski}. In the Kraus representation, a quantum operation is represented by a map of the form $\Phi(\rho) = \sum_i K_i \rho K_i^\dag$ which is completely specified by a set of operators $\{K_i\}$ called Kraus Operators. The Kraus operators must satisfy the completeness relation $\sum_i K^\dag_i K_i = \openone$ in order to qualify as a valid quantum operation. In the Choi-Jamiolkowski representation, a quantum map $\Phi$ is represented by an operator $J(\Phi) = \sum_{i,j}\Phi(\ket{i}_A\bra{j}) \otimes \ket{i}_B\bra{j}$ which satisfies $\Tr_A[J(\Phi)] = \openone_B$. The action of $\Phi$ on some state $\rho$ is then recovered via the map $\Tr_B[J(\Phi) \openone_A \otimes \rho_B^T] = \Phi(\rho_A)$. A simple relationship connects both equivalent representations. For a map $\Phi$ represented by Kraus operators $\{K_i = \sum_{j,k}K_{i,j,k} \ket{j}\bra{k}\}$, the corresponding Choi-Jamiolkowski representation is $J(\Phi) = \sum_{i} v_i v^\dag_i$ where $v_i \coloneqq \sum_{j,k}K_{i,j,k} \ket{j}\ket{k}$.

The notion of coherence that we will employ in this paper will be the one identified in~\cite{Aberg2006, Baumgratz2014}, where a set of axioms are identified in order to specify a reasonable measure of quantum coherence. The axioms are as follows:

For a given fixed basis $\{ \ket{i} \}$, the set of incoherent states $\cal I$ is the set of quantum states with diagonal density matrices with respect to this basis. Incoherent completely positive and trace preserving maps (ICPTP)
 are maps that map every incoherent state to another incoherent state. Given this, we say that  $\mathcal{C}$ is a measure of quantum coherence if it satisfies following properties:
(C1) $\mathcal{C}(\rho) \geq 0$ for any quantum state $\rho$ and equality holds if and only if $\rho \in \cal I$.
(C2a) The measure is non-increasing under a ICPTP map $\Phi$ , i.e., $C(\rho) \geq C(\Phi(\rho))$.
(C2b) Monotonicity for average coherence under selective outcomes of ICPTP:
$C(\rho) \geq \sum_n p_n C(\rho_n)$, where $\rho_n = K_n \rho K_n^\dagger/p_n$ and $p_n = \Tr [K_n \rho K^\dagger_n ]$ for all $K_n$ with $\sum_n K_n K^\dagger_n = \mathbb 1$ and $K_n {\cal I} K_n^\dagger \subseteq \cal I$.
(C3) Convexity, i.e. $\lambda C(\rho) + (1-\lambda) C(\sigma) \geq C(\lambda \rho + (1-\lambda) \sigma)$, for any density matrix $\rho$ and $\sigma$ with $0\leq \lambda \leq 1$.

One may check that a particular operation is incoherent if its Kraus operators always maps a diagonal density matrix to another diagonal density matrix. One important example of such an operation is the CNOT gate. We can also additionally distinguish between the maximal set of ICPTP maps, which from now on we refer to as maximally incoherent operations (MIO) from the set of ICPTP maps whose Kraus operators that additionally satisfy $K_n {\cal I} K_n^\dagger \subseteq \cal I$, which we refer to as simply incoherent operations (IO) . From this definition, it is clear that $\mathrm{IO} \subset \mathrm{MIO}$.

Finally, we review some basic notions regarding semidefinite programs. A semidefinite program is a linear optimization problem over the set of positive matrices $X$, subject to a set of constraints that can be expressed in the following form:

\begin{equation*}
\begin{aligned}
& \underset{X \geq 0}{\text{max}}
& & \Tr(AX) \\
& \text{subject to}
& & \phi_i (X) = B_i, \; i = 1, \ldots, m. 
\end{aligned}
\end{equation*}                                              

where $A$ and $B_i$ are Hermitian matrices and $\phi_i$ is a linear, Hermiticity preserving map (i.e. it maps every Hermitian matrix to another Hermitian matrix) representing the $i$th constraint. The above is called the primal problem. The optimal solution to the primal problem is always upper bounded by the optimal solution to the dual problem, when they exist. The dual problem may be written as the following optimization problem over all possible Hermitian matrices $Y_i$:

\begin{equation*}
\begin{aligned}
& \underset{\{Y_i = Y_i^\dag\}}{\text{min}}
& & \sum_{i=1}^m \Tr(B_i Y_i) \\
& \text{subject to}
& & \sum_{i=1}^m \phi^*_i (Y_i) \geq A. 
\end{aligned}
\end{equation*}   

In this case, $\phi_i^*$ refers to the conjugate map that satisfies $\Tr[C^\dag \phi_i (D)] = \Tr[ \phi^*_i (C)^\dag D]$ for every matrix $C$ and $D$.

In fact, the solutions to the primal and dual problems are almost always equal except in the most extreme cases. Nonetheless, this still needs to be verified on a case by case basis. A sufficient condition for both primal and dual solution to be equal is called Slater's Theorem, which states that if the set of positive matrices $X$ that satisfies all the constraints $\phi_i$ is nonempty, and if the set of Hermitian matrices $\{ Y_i \}$ that satisfies the strict inequality $\sum_{i=1}^m \phi^*_i (Y_i) > A$ is also nonempty, then the optimal solutions for both problems, also referred to as the the optimal primal value and the optimal dual value, must be equal.

\section{Coherence measures from maximally incoherent operations.}
\label{coherence measures}

\begin{theorem} \label{thm::strongMono}
For any quantum observable $M$ and quantum state $\rho$, the quantity $$ \max_{\Phi \in \mathcal{O}}\mathrm{Tr}(M \Phi(\rho))$$ is strongly monotonic under incoherent operations, where $\mathcal{O}$ may be substituted with either the set of operations MIO or IO.
\end{theorem}

\begin{proof}
We first observe that any incoherent operation represented by some set of incoherent Kraus operators $\{ K^{\mathrm{IO}}_i \}$ is, by definition, also a maximally incoherent operation. Note that for any set of maximally incoherent operations $\{ \Omega_i^{\mathrm{MIO}} \mid \Omega_i^{\mathrm{MIO}} \in \mathrm{MIO} \}$, the map $\Omega(\rho) \coloneqq \sum_i \Omega_{i}^{\mathrm{MIO}}(K^{\mathrm{IO}}_i\rho K^{\mathrm{IO} \dag}_i )$ is also maximally incoherent since it is just a concatenation of the incoherent operation represented by $\{ K^{\mathrm{IO}}_i \}$, followed by performing a maximally incoherent operation $\Omega_i^{\mathrm{MIO}}$ conditioned on the measurement outcome $i$. Let us assume that $\Omega_i^{\mathrm{MIO}}(\rho_i)$
 is the optimal maximally incoherent operation maximizing $\mathrm{Tr}(M \Omega_i^{\mathrm{MIO}}(\rho_i))$ for the state $\rho_i \coloneqq K_i^{\mathrm{IO}}\rho K_i^{\mathrm{IO}\dag} /\Tr (K_i^{\mathrm{IO}}\rho  K_i^{\mathrm{IO}\dag})$, we then have the following series of inequalities:

\begin{align*}
\max_{\Phi \in MIO}\mathrm{Tr}(M \Phi(\rho)) &\geq \mathrm{Tr}(M \Omega(\rho)) \\
& = \mathrm{Tr}[M \sum_i \Omega_i^{\mathrm{MIO}}(K_i^{\mathrm{IO}}\rho K_i^{\mathrm{IO}\dag})] \\
& = \mathrm{Tr}[M \sum_i  p_i\Omega_i^\mathrm{MIO}(\rho_i)] \\
& = \sum_i p_i \max_{\Phi_i \in MIO}\mathrm{Tr}(M \Phi_i(\rho_i)),
\end{align*} where $\rho_i \coloneqq K_i^{\mathrm{IO}}\rho K_i^{\mathrm{IO}\dag} /\Tr (K_i^{\mathrm{IO}}\rho  K_i^{\mathrm{IO}\dag})$ and $p_i \coloneqq \Tr (K_i^{\mathrm{IO}}\rho  K_i^{\mathrm{IO}\dag)})$. We note that the last line is simply the expression for strong monotonicity, which proves the result for the case when $\mathcal{O}$ is MIO. Identical arguments apply when considering IO, which completes the proof.
\end{proof}

In the above, we see that the optimization over MIO in fact yields a valid coherence monotone in within the regime of IO, so in fact, drawing a sharp distinction between the two sets of operations is not always necessary.

We note that satisfying strong monotonicity qualifies the quantity to be considered a coherence monotone, but is insufficient to qualify it to be considered as a coherence measure. In order for that to happen, we still need to demonstrate that $\max_{\Phi \in \mathcal{O}}\mathrm{Tr}(M \Phi(\rho)) = 0$ iff $\rho$ is an incoherent state, and $\max_{\Phi \in \mathcal{O}}\mathrm{Tr}(M \Phi(\rho)) > 0$ whenever $\rho$ is a coherent state. It is clear that this is only true for some special cases of $M$. However, the following theorem shows that even if $M$ does not by itself satisfy the above condition, it is still always possible to construct a valid coherence measure using $M$.

\begin{theorem}\label{thm::measure}
Let M be some Hermitian quantum observable with a complete set of eigenstates denoted by $\{\ket{\alpha_i}\}$. Then there always exists some basis $\{ \ket{i} \}$ such that $\bra{i}(M - \frac{\Tr{M}}{d}\openone)\ket{i} = 0$ for every $\ket{i}$ where $d$ is the dimension of the Hilbert space.

Furthermore, for every nontrivial quantum observable $M$, the quantity $$\mathcal{C}^\mathcal{O}(\rho) \coloneqq \max_{\Phi \in \mathcal{O}}\mathrm{Tr}[( M\Phi(\rho)] - \Tr(M)/d$$ is always a valid coherence measure w.r.t. the basis $\{ \ket{i} \}$. The set of quantum maps $\mathcal{O}$ may be subtituted with either MIO or IO.
\end{theorem}

\begin{proof}
We begin by observing that the matrix $M' = M - \frac{\Tr{M}}{d}\openone$ is trace zero. Since $M'$ is nontrivial, this implies that the sum of its positive eigenvalues and negative eigenvalues must be exactly equal. Let $\vec{\lambda} = (\lambda_1, \ldots , \lambda_d)$ be the vector of eigenvalues of $M'$ arranged in decreasing order. We recall the Schur-Horn theorem, which states that for every vector $\vec{v} = (v_1, \ldots, v_d)$, there exists a Hermitian matrix with the same vector of eigenvalues $\vec{\lambda}$, but with diagonal entries $\vec{v} = (v_1, \ldots, v_d)$ so long as the vectors satisfy the majorization condition $  \vec{v}  \prec \vec{\lambda} $. It is clear that the zero vector $\vec{v} = (0, \ldots, 0)$ always satisfies this condition. Therefore, there always exist a basis $\{ \ket{i} \}$ for $M'$ where the main diagonals are all zero, such that $\bra{i}M'\ket{i} = 0$ for every $\ket{i}$, which proves the first part of the theorem.

Now, we proceed to prove that $\mathcal{C}^{\mathcal{O}}_M(\rho)$ is a coherence measure of with respect to the basis $\{ \ket{i} \}$. The strong monotonicity condition is already satisfied due to Thm~\ref{thm::strongMono}. The convexity of the measure is immediate from the linearity of the trace operation and the definition of $\mathcal{C}^\mathcal{O}_M$ as a maximization over MIO or IO. Therefore, we only need to establish the faithfulness property of the measure.

In order to prove this, recall that in the basis $\{ \ket{i} \}$, the diagonal elements of $M'$ is all zero. Therefore, there always exists some projection onto a 2 dimensional space $M'$ such that the corresponding submatrix has the form $\begin{pmatrix}
    0 & r  \\
     r*  & 0
\end{pmatrix}$. We can assume without loss of generality that the projection is onto the subspace $\{ \ket{0},\ket{1} \}$, since at this point, the numerical labelling of the basis is arbitary.

For some coherent quantum state $\rho$, there is at least one nonzero off-diagonal element. Since basis permutation is an incoherent operation, we can assume the nonzero off-diagonal element is $\rho_{01}$. In fact, we can assume that it is the only nonzero off diagonal element as we ca freely  project onto the subspace spanned by $\{ \ket{0},\ket{1} \}$ and completely dephase the rest of the Hilbert space via an incoherent operation, which allows us to prove the general result by only considering the 2 dimensional case. Suppose this leads to a 2 dimensional submatrix of the form $\begin{pmatrix}
    p_1 & a  \\
     a^*  & p_2
\end{pmatrix}$ where $a$ is nonzero since $\rho$ is coherent. 


Directly computing $\Tr\begin{pmatrix}
    0 & r  \\
     r^*  & 0
\end{pmatrix}\begin{pmatrix}
    p_1 & a  \\
     a^*  & p_2
\end{pmatrix}$, we get the expression $r^*a+a^*r = \abs{ra}(e^{i\phi} + e^{-i\phi}) $. This final quantity  can always be made positive by performing the incoherent unitary that performs $\ket{0} \rightarrow \ket{0}$ and  $\ket{1} \rightarrow e^{-i\phi}\ket{1}$ which is equivalent to making both $a$ and $r$ positive quantities. Since $r$ is strictly positive as $M'$ is a nontrivial matrix, this implies $ar > 0$ if $\rho$ is a coherent state, so there always exists at least one incoherent operation $\Phi$ such that $\Tr[M'\Phi(\rho)] > 0$ for every coherent state $\rho$. 

Finally, we just observe that $M'$ has zero diagonal elements w.r.t. the basis $\{ \ket{i}\}$, so $\Tr[M'\Phi(\rho)]=0$ whenever $\rho$ is incoherent and $\Phi$ is MIO or IO. This completes the proof. 
\end{proof}

Theorem~\ref{thm::measure} above establishes several facts. First, observe that since $\mathcal{C}^\mathcal{O}(\rho)$ is a coherence measure and nonnegative, $\mathrm{Tr}[( M\rho)] - \Tr(M)/d$ can only be positive when $\rho$ is coherent (the basis is specified by the theorem). This establishes that every nontrivial observable $M$ is in fact a witness of \textit{some} form of coherence. One just needs to subtract the constant $\Tr(M)/d$ from the mean value $\langle M \rangle$. 

Second, it establishes that if $M$ is a coherence witness, then it can be interpreted as the lower bound of the \textit{bona fide} coherence measure $\mathcal{C}^\mathcal{O}_M$. Recall that the measure $\mathcal{C}^\mathcal{O}_M$ quantifies the operational usefulness of a quantum state when one considers MIO or IO type quantum operations and the task is to maximize a given  observable $M$. Other examples of coherence measures with operational interpretations in terms of MIO or IO include the relative entropy of coherence, which quantifies the number of maximally coherent qubits you can distill using IO, as quantities considering how much  entanglement and Fisher information can be extracted via MIO or IO.

Third, Theorem~\ref{thm::measure} exactly specifies the preferred basis that is useful for optimizing $M$ and that such a basis always exists. The following proposition that coherence within any mutually unbiased bases will always satisfy the necessary condition.

\begin{proposition} \label{thm::mub}
Let $\{ \ket{\alpha_i} \}$ be the complete set of eigenbases of some nontrivial quantum observable M, and let $\{ \ket{\beta_i} \}$ be any complete basis that is mutually unbiased w.r.t. $\{ \ket{\alpha_i} \}$. Then the basis $\{ \ket{\beta_i} \}$ always satisfies $\bra{\beta_i}(M - \frac{\Tr{M}}{d}\openone)\ket{\beta_i} = 0$ for every $\ket{\beta_i}$. 

In other words, w.r.t. any mutually unbiased basis $\ket{\beta_i}$, the diagonal elements of $M - \frac{\Tr{M}}{d}\openone$ is always zero.
\end{proposition}

\begin{proof}
Let the dimension of the Hilbert space be $d$. We then have $\abs{\bra{\beta_i}\alpha_j\rangle}^2 = \frac{1}{d}$. Since $\{ \ket{\alpha_i} \}$ is the complete eigenbasis of $M$, $M= \sum_{i} \lambda_i \ket{\alpha_i}\bra{\alpha_i}$ and $\bra{\beta_i}M\ket{\beta_i} = \sum_j \frac{\lambda_j}{d} = \frac{\Tr{M}}{d}$. This implies that $\bra{\beta_i}(M - \frac{\Tr{M}}{d}\openone)\ket{\beta_i} = 0$ for every $i = 1, \ldots, d$, which is the required condition.
\end{proof}

\section{A semidefinite program for computing coherence measures}
\label{sdp}

Previously, we have considered both MIO and IO during the construction of our coherence measures. Here, we show that for MIOs, the corresponding coherence measure $\mathcal{C}^{\mathrm{MIO}}_M$ is in fact, efficiently computable via a semidefinite program.

Let us first define the matrix $A \coloneqq M_A \otimes \rho_B^T \otimes \ket{1}_C\bra{1} $ acting on $\mathcal{H}_A \otimes \mathcal{H}_B \otimes \mathcal{H}_C \otimes \mathcal{H}_D$. Furthermore, we will assume that $\mathrm{dim}(\mathcal{H}_A) = \mathrm{dim}(\mathcal{H}_B) = \mathrm{dim}(\mathcal{H}) = d$ and $\mathrm{dim}(\mathcal{H}_D) = 2$.

We now prove the following:

\begin{theorem} \label{thm::primal}
For any quantum observable M, the optimization problem $$\max_{\Phi \in \mathrm{MIO}} \Tr(M\Phi(\rho))$$ is equivalent to the semidefinite program

\begin{equation*}
\begin{aligned}
& \underset{X \geq 0}{\text{max}}
& & \Tr(AX) \\
& \text{subject to}
& & \Tr_{AC}(X\ket{1}_C\bra{1}) = \openone_B \\
& & &\Tr_{BC}(X \openone_A\otimes \ket{i}_B\bra{i} \otimes \ket{1}_C\bra{1} ) \\
&&&=  \sum_{j=1}^d\Tr_{ABC}(X \ket{j}_A\bra{j} \otimes \ket{i}_B\bra{i} \otimes \ket{2}_C\bra{2} ) \ket{j}_A\bra{j}\\
& & & \forall \; i = 1, \ldots, d, 
\end{aligned}
\end{equation*} where $A \coloneqq M_A \otimes \rho_B^T \otimes \ket{1}_C\bra{1} $. 
\end{theorem}

\begin{proof}
We begin by first noting that the matrix $X$ can be written as the matrix $$\begin{pmatrix}
    X_1 & *  \\
     *  & X_2
\end{pmatrix}. $$ The $*$ indicates possible nonzero elements, but they do not appear in the objective function we are trying to optimize, nor do they appear within the linear constraints, so they can be arbitrary so long as $X\geq 0$. The matrix $A$ written in matrix form looks like $$\begin{pmatrix}
    M_A \otimes \rho_B^T & 0  \\
     0  & 0
\end{pmatrix}. $$ Computing $\Tr(AX)$, we get $$\Tr(AX) = \Tr_{A}[\Tr_B(X_1 \openone_{A} \otimes \rho_B^T) M_A ]. $$

Now, the constraint $\Tr_{AC}(X\ket{1}_C\bra{1}) = \openone_B$ implies $\Tr_{A}(X_1) = \openone_B$, so $X_1$ actually represents a valid quantum operation in the Choi-Jamiolkowski representation. This implies $\Tr(AX)$ has the form $\Tr_A[\Phi(\rho) M_A ]$ for some valid quantum operation $\Phi$.

All that remains is for us to prove that under the set of constraints \begin{align*}
\Tr_{BC}(X \openone_A \otimes \ket{i}_B\bra{i} \otimes & \ket{1}_C\bra{1}  ) \\
&=  \sum_{j}\Tr_{ABC}(X \ket{j}_A\bra{j} \otimes \ket{i}_B\bra{i}  \otimes \ket{2}_C\bra{2} ) \ket{j}_A\bra{j}
\end{align*} for all  $i = 1, \ldots, d$ and  $j = 1, \ldots, d$, $\Phi$ must be a maximally incoherent operation. We first note that the number $\Tr_{ABC}(X \ket{j}_A\bra{j} \otimes \ket{i}_B\bra{i}  \otimes \ket{2}_C\bra{2} )$ is just the main diagonal elements of the matrix $X_2$, so it must be nonnegative since $X$ is positive and $X_2$ is a principle submatrix of $X$. We can therefore rewrite the constraint as $\Tr_{BC}(X \openone_A \otimes \ket{i}_B\bra{i} \otimes  \ket{1}_C\bra{1}  )
=  \sum_j \lambda_{i,j} \ket{j}_A\bra{j}$ where $\lambda_{i,j} $ is nonnegative. This necessarily means that every incoherent state $\ket{i}\bra{i}$ is mapped to a diagonal state $\sum_j \lambda_{i,j} \ket{j}\bra{j}$ under the quantum map represented by $X_1$, which defines maximally incoherent operations, and completes the proof.

\end{proof}

Given the primal problem in Theorem~\ref{thm::primal}, we can also write down the dual problem, which is detailed in the following corollary:

\begin{corollary}
The dual to the primal problem in Theorem~\ref{thm::primal} is the following optimization over all possible Hermitian $Y_A$ and $Y_B$:

\begin{equation*}
\begin{aligned}
& \underset{Y_B = Y_B^\dag}{\text{min}}
& &  \Tr( Y_B) \\
& \text{subject to}
& & \openone_{A} \otimes Y_B + Y_A\otimes \openone_B \geq M_A \otimes \rho_B^T \\
& & &   \bra{j}Y_A \ket{j}_A \leq 0, \; \forall j= 1, \ldots, d\\
& & &  \\
\end{aligned}
\end{equation*}

Furthermore, the optimal primal value is equal to the optimal dual value.                                                                                                                                       
\end{corollary}

\begin{proof}
The first constraint in the primal problem can be written as $ \phi(X) \coloneqq \Tr_{AC}(X\ket{1}_C\bra{1}) = \openone_B$. The conjugate map can be verified to be the map $\phi^*(Y_B) = \openone_{A} \otimes Y_B \otimes \ket{1}_C\bra{1}$, since it satisfies $\Tr[Y_B \phi(X)] = \Tr[\phi^*(Y_B)X]$.

The rest of the constraints can be written as $$\begin{aligned} \phi_{i}(X) \coloneqq &\Tr_{BC}(X \openone_A \otimes \ket{i}_B\bra{i} \otimes  \ket{1}_C\bra{1}  ) \\
&-  \sum_{j}\Tr_{ABC}(X \ket{j}_A\bra{j} \otimes \ket{i}_B\bra{i}  \otimes \ket{2}_C\bra{2} ) \ket{j}_A\bra{j} = 0. \end{aligned}$$ In this case the conjugate map is  $$\begin{aligned} \phi^*_{i}(Y_A^{i}) \coloneqq & Y_A^{i}\otimes \ket{i}_B\bra{i}  \otimes \ket{1}_C\bra{1}\\ &- \sum_j\bra{j}Y_A^{i} \ket{j}_A\ket{j}_A\bra{j} \otimes \ket{i}_B\bra{i} \otimes \ket{2}_C\bra{2}. \end{aligned}$$

Summing over the variable $i$, we have $$\begin{aligned} \sum_i\phi^*_{i}(Y_A^{i}) \coloneqq & \sum_i Y^i_A\otimes \ket{i}_B\bra{i} \otimes  \ket{1}_C\bra{1}\\ &- \sum_{i,j} \bra{j}Y^i_A \ket{j}_A\ket{j}_A\bra{j} \otimes \ket{i}_B\bra{i} \otimes \ket{2}_C\bra{2}. \end{aligned}$$.

The dual program can therefore be written as:

\begin{equation*}
\begin{aligned}
& \underset{Y_B = Y_B^\dag}{\text{min}}
& &  \Tr( Y_B) \\
& \text{subject to}
& & \openone_{A} \otimes Y_B \otimes \ket{1}_C\bra{1} + Y_A\otimes \openone_B \otimes  \ket{1}_C\bra{1} - \\
& & &  \sum_i Y^i_A\otimes \ket{i}_B\bra{i} \otimes  \ket{1}_C\bra{1} - \\
& & &  \sum_{i,j} \bra{j}Y^i_A \ket{j}_A\ket{j}_A\bra{j} \otimes \ket{i}_B\bra{i} \otimes \ket{2}_C\bra{2}\\
& & & \geq M_A \otimes \rho_B^T \otimes \ket{1}_C\bra{1} \\
\end{aligned}
\end{equation*}

The third line of the constraint is actually just $- \sum_{i,j} \bra{j}Y^i_A \ket{j}_A\ket{j}_A\bra{j} \otimes \ket{i}_B\bra{i} \otimes \ket{2}_C\bra{2} \geq 0$, which is equivalent to the contraint that the main diagonal of $Y^i_A$ is all negative. As such, the program can be further simplified to the following:

\begin{equation*}
\begin{aligned}
& \underset{Y_B = Y_B^\dag}{\text{min}}
& &  \Tr( Y_B) \\
& \text{subject to}
& & \openone_{A} \otimes Y_B + \sum_i Y^i_A\otimes \ket{i}_B\bra{i} \geq M_A \otimes \rho_B^T \\
& & &   \bra{j}Y^i_A \ket{j}_A \leq 0, \; \forall j= 1, \ldots, d\\
& & &  \\
\end{aligned}
\end{equation*}

which is the form that was presented in the corollary. Finally, we just need to check that the primal and dual programs satisfies Slater's conditions. For the primal problem, the optimization is over all MIO's, so the primal feasible set is nonempty (for instance, we can just consider the Choi-Jamiolkowski representation of the identity operation, which also falls under MIO). Furthermore, there exists at least one set of $Y^i_A$ and $Y_B$ s.t. $\openone_{A} \otimes Y_B + \sum_i Y^i_A\otimes \ket{i}_B\bra{i} > M_A \otimes \rho_B^T$ since one can always set $Y^i_A = 0$, and $Y_B = x \openone_B$ where $x > \lambda_{\max}(M_A \otimes \rho_B^T)$ and $\lambda_{\max}(A)$ represents the largest eigenvalue of $A$. As such, Slater's conditions are satisfied and the primal optimal value is equal to the dual optimal value.

\end{proof}

\section{Relation to robustness and $l_1$ norm of coherence}

It was observed in~\cite{Napoli2016} that the robustness of coherence $\mathcal{C}_\mathcal{R}$, which may be interpreted as the minimal amount of quantum noise that can be added to a system before it becomes incoherent, is a coherence measure that is also simultaneously an observable. That is, for any state $\rho$, there always exists some optimal witness $W_\rho$ such that  $\Tr( W_\rho \rho) = \mathcal{C}_\mathcal{R}(\rho)$. It was also demonstrated that the $l_1$ norm upper bounds the robustness, so $\mathcal{C}_\mathcal{R}(\rho) \leq \mathcal{C}_{l_1}(\rho)$. The following theorem shows that both the robustness and the $l_1$ norms of coherence are fundamental upper bounds of  $\mathcal{C}^\mathcal{O}_M$. We note that in \cite{Ren2018}, it was also observed that when $M$ is a witness that achieves its maximum value for the maximally coherent state, then $\mathcal{C}^{\mathrm{IO}}_M$ is upper bounded by the $l_1$ norm of coherence under certain normalization conditions.

\begin{theorem} [Hierarchy of coherence measures] \label{thm::hierarchy}
For any given state $\rho$ and observable $M$, the following hierarchy of the coherence measures holds: $$\mathcal{C}^{\mathrm{IO}}_M(\rho) \leq \mathcal{C}^{\mathrm{MIO}}_M(\rho) \leq \mathcal{N}_M \mathcal{C}_\mathcal{R}(\rho) \leq \mathcal{N}_M\mathcal{C}_{l_1}(\rho) $$ where $\mathcal{N}_M \coloneqq \abs{\lambda_\mathrm{min}(M)- \frac{\Tr M}{d}}$ and $\lambda_\mathrm{min}(M)$  is the smallest eigenvalue of the observable $M$. Furthermore, all the inequalities are tight.
\end{theorem}

\begin{proof}
In~\cite{Napoli2016}, it was shown that $\mathcal{C}_\mathcal{R}(\rho)$ is equivalent to maximizing $\Tr{\rho W}$ over all Hermitian observables $W$, subject to the constraint that $W \geq - \openone $ and that the diagonal entries of $W$ are nonnegative. Note that our convention differs from the one presented in \cite{Napoli2016} by a negative sign.

We always displace $M$ and consider the matrix $M' = M - \frac{\Tr M}{d} \openone$, and it is clear that a positive scaling factor does not fundamentally change $\mathcal{C}^\mathcal{O}_M$ where $\mathcal{O}$ is MIO or IO. As such, without any loss in generality, we can assume that $M$ is a traceless matrix where the leading matrix elements are zero, and that its smallest eigenvalue is normalized such that $\lambda_\mathrm{min}(M) = -1$. This implies that $\mathcal{N}_M = 1$.  Observe that under these assumptions, $M$ automatically satisfies the constraints on $W$ that was described in the preceding paragraph. 

Recall that $\mathcal{C}^\mathcal{O}_M (\rho) \coloneqq \max_{\Phi \in \mathrm{MIO}} \Tr(M\Phi(\rho))$. Consider the quantity $\Tr(M\Phi(\rho))$ and let $\Phi^*$ be the conjugate map such that $\Tr(M\Phi(\rho)) =\Tr(\Phi^*(M)\rho)$. Since $\Phi$ is a CPTP map and the conjugate map preserves the trace, it cannot decrease the minimum eigenvalue so $\lambda_\mathrm{min}[\Phi^*(M)] \geq \lambda_\mathrm{min}(M)$. Furthermore, we see that as $\Phi$ is MIO or IO, the leading diagonals of $\Phi^*(M)$ must be zero if the leading diagonals of $M$ are zero. This again comes from the definition of the conjugate map $\Tr(M\Phi(\rho)) =\Tr(\Phi^*(M)\rho)$. From this, we can determine that $\Phi^*(M)$ always satisfies the necessary constraints for $W$ described above, and this is true for any $\Phi$ that is an incoherent CPTP map, so we must have  $\mathcal{C}^\mathcal{O}_M(\rho) \leq \mathcal{N}_M \mathcal{C}_\mathcal{R}(\rho)$. 

It was already known that $\mathcal{C}_\mathcal{R}(\rho) \leq \mathcal{C}_{l_1}(\rho)$, and we must have that $\mathcal{C}^{\mathrm{IO}}_M(\rho) \leq \mathcal{C}^{\mathrm{MIO}}_M(\rho)$ since $\mathrm{IO} \subset \mathrm{MIO}$, which leads to the final chain of inequalities $$\mathcal{C}^{\mathrm{IO}}_M(\rho) \leq \mathcal{C}^{\mathrm{MIO}}_M(\rho) \leq \mathcal{N}_M \mathcal{C}_\mathcal{R}(\rho) \leq \mathcal{N}_M\mathcal{C}_{l_1}(\rho). $$

To see that the inequalities are in fact tight, we need to demonstrate that there are cases of $M$ and $\rho$ where equality is achieved. It is already known that when the dimension of the system is $d=2$ then the robustness is identical to the $l_1$ norm of coherence~\cite{Napoli2016}. Furthermore, we know that for any $\rho$ there always exists $W_\rho$ where $\Tr( W_\rho \rho) = \mathcal{C}_\mathcal{R}(\rho)$. In this case, we can simply choose $M=W_\rho$, which is enough to achieve $\mathcal{C}^{\mathrm{MIO}}_M(\rho) = \mathcal{N}_M \mathcal{C}_\mathcal{R}(\rho)$. Finally, we can verify that $\mathcal{C}^{\mathrm{IO}}_M(\rho) = \mathcal{C}^{\mathrm{MIO}}_M(\rho)$ is achieved when the input state $\rho$ is the maximally coherent state and we choose $M = \rho$. Therefore, all the inequalities are tight.
\end{proof}

\begin{figure}
	\centering
    \includegraphics[width=1\linewidth]{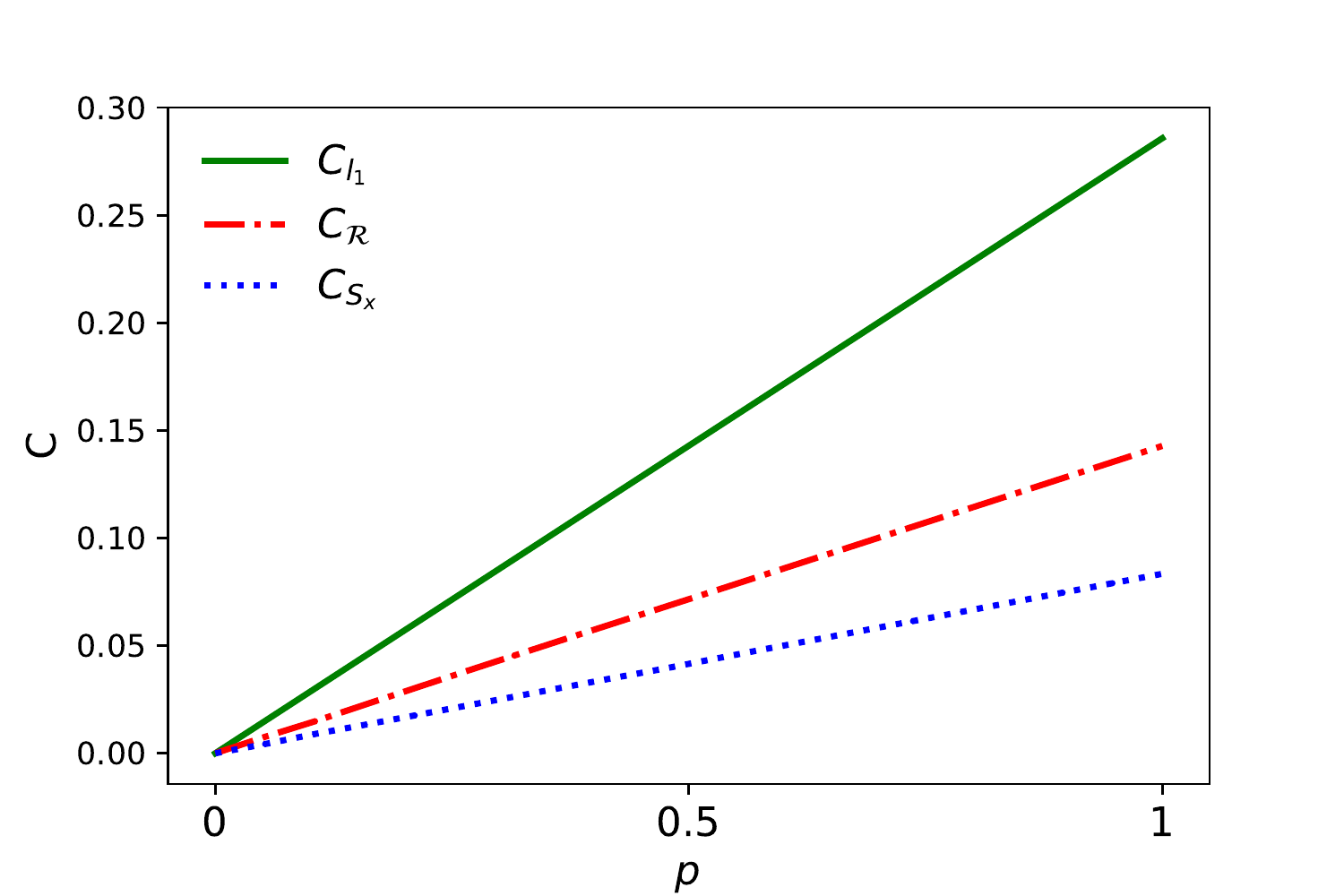}
    \caption{Comparisons among the $l_1$ norm of coherence $\mathcal{C}_{l_1}$(green, solid), the robustness of coherence $\mathcal{C}_{\mathcal{R}}$(red, dash-dotted), and the coherence measure corresponding to magnetization measurement $\mathcal{C}_{S_x}$(red, dotted). We consider the single parameter, 3 qubit state $\rho = (1+p/7)\mathbb{1}/8 - p/7\ket{w}\bra{w}$, where $\ket{w} \coloneqq \frac{1}{\sqrt{3}}(\ket{001}+\ket{010}+\ket{100})$ and $p \in [0,1]$}
    \label{fig::spin}
\end{figure}

\begin{figure}
	\centering
   	\includegraphics[width=1\linewidth]{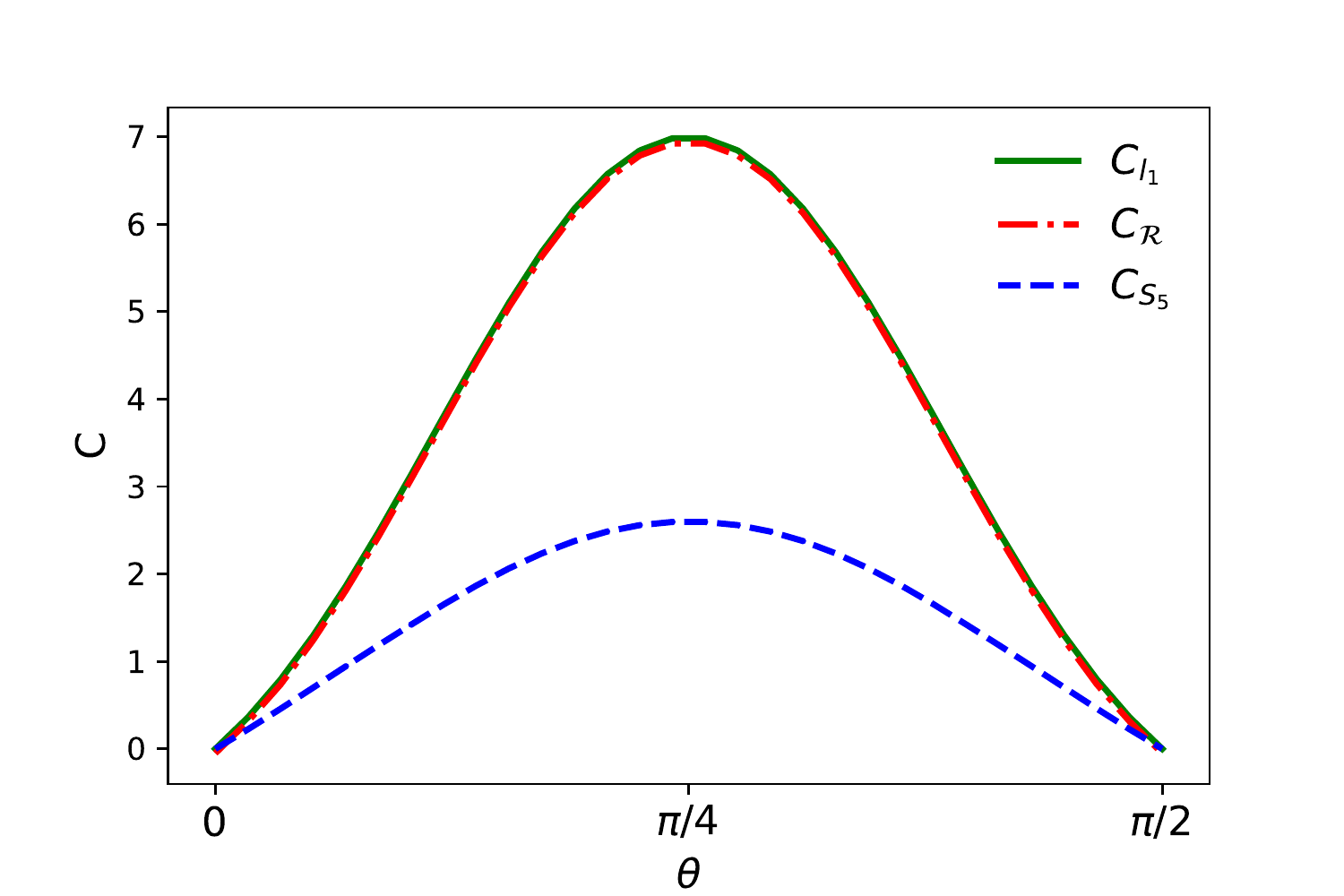}
    \caption{A comparison of $\mathcal{C}_{S_5}$(blue, dashed) with $\mathcal{C}_{l_1}$(green, solid) and $\mathcal{C}_{\mathcal{R}}(red, dash-dotted)$for the state $\ket{\psi(\theta)} = (\cos(\theta)\ket{g} + \sin(\theta)\ket{e})^{\otimes 3}\ket{0}^{\otimes 2}$. Note that a 5 qubit system (as opposed to a 3 qubit one) was chosen in order to avoid saturation of the measure $\mathcal{C}_{S_N}$ when $N=3$. While the quantities are different, the qualitative behaviours are similar across $\theta \in [0,\frac{\pi}{2}]$. For pure state, $\mathcal{C}_{l_1}$ and $\mathcal{C}_\mathcal{R}$ matches in general.}
    \label{fig::superradiance}
\end{figure}

\section{Examples}\label{Examples}
In this section we present numerical examples of coherence measure of the computable measure $\mathcal{C}^{\mathrm{MIO}}_{M}$.

Let us consider for spin systems the total magnetic moment operator. For a system of $N$  spins we can choose for our classical basis $\bigotimes_{i = 1} ^{N} \{ \ket{\uparrow}_i, \ket{\downarrow}_i \}$ where $\{ \ket{\uparrow}_i, \ket{\downarrow}_i \}$ is the eigenbasis of the local spin-$z$ operator. In order to witness the coherence between these basis states, then a simple measurement of the magnetization in the $x$ direction will suffice (See Theorem~\ref{thm::measure} as well as Proposition~\ref{thm::mub}).  The total spin-$x$ operator is defined as $$S_x = \sum_{i = 1}^{N} S_x^i$$ with local spin operators $S_x^i$. Choosing $S_x$ as our observable, any measurement of $\langle S_x \rangle$ is automatically a lower bound to the corresponding coherence measure $\mathcal{C}_{S_x}^\mathcal{O}$. Note that because one can equivalently choose to measure the total magnetization along any direction on the equatorial plane, any non zero measurement of $\langle S_x \rangle$ directly implies the presence of coherence in the $z$ direction. 

One may also choose to instead find the `optimal' measure by finding then implementing the optimal observable achieving $\Tr( W_\rho \rho) = \mathcal{C}_\mathcal{R}(\rho)$~\cite{Napoli2016}. However, the physical implementation of such an observable $W_\rho$ is not always simple. Moreover, if one were interested to quantify the total coherence in the system, there is also no computational advantage to finding the robustness since both $\mathcal{C}^{\mathrm{MIO}}_{S_x}$ and $\mathcal{C}_\mathcal{R}$ are computable via semidefinite programs. This example neatly illustrates how the resource requirements for experimentally detecting and measuring quantum coherence may be simplified via the direct application of Theorem~\ref{thm::measure} and Proposition~\ref{thm::mub}. Figure~\ref{fig::spin} compares $\mathcal{C}_{l_1}$, $\mathcal{C}_\mathcal{R}$ and $\mathcal{C}^{\mathrm{MIO}}_{S_{x}}$ for the state $\rho = (1+p/7)\mathbb{1}/8 - p/7\ket{w}\bra{w}$ where $\ket{w} \coloneqq \frac{1}{\sqrt{3}}(\ket{001}+\ket{010}+\ket{100})$ and $p \in [0,1]$. Note the hierarchy of the coherence measures $\mathcal{C}^{\mathrm{MIO}}_M(\rho) \leq \mathcal{N}_M \mathcal{C}_\mathcal{R}(\rho) \leq \mathcal{N}_M\mathcal{C}_{l_1}(\rho)$ (See Theorem~\ref{thm::hierarchy}).

Several existing coherence measures can also be shown to fall under the framework that was discussed in this article. For instance, in~\cite{Tan2018}, superradiance is studied within the context of coherence. In the idealized model for superradiance, there are $N$-number of two-level atomic systems with the energy levels denoted by $\ket{e^{(i)}}$ and $\ket{g^{(i)}}$ respectively. From this, we define the raising and lowering operators acting on the $i$th subsystem as $D^{(i)}_+ \coloneqq \ket{e^{(i)}}\bra{g^{(i)}}$ and $D^{(i)}_- \coloneqq \ket{g^{(i)}}\bra{e^{(i)}}$, and the collective component of the emission rate, referred to as the superradiant quantity, is $\langle S_N \rangle = \sum_{i\neq j} \langle D_+^{(i)} D_-^{(j)} \rangle$. We see that $S_N$ is a traceless observable whose leading diagonal elements are all zero in the axis defined by $\ket{e^{(i)}}$ and $\ket{g^{(i)}}$. This neatly falls underneath our framework, so any witnessing of superradiance is in fact, a witness of coherence between these basis states and a computable measure $\mathcal{C}^{\mathrm{MIO}}_{S_N}$ may be constructed. We note that this is a considerable improvement upon the original measure in~\cite{Tan2018}, which uses the computationally difficult convex roof construction in order to generalize the measure to a general mixed state.  A comparison of $\mathcal{C}_{S_N}^{\mathrm{MIO}}$ with other coherence measures for the pure state $\ket{\psi(\theta)} = (\cos(\theta)\ket{g} + \sin(\theta)\ket{e})^{\otimes 3}$ is shown in Figure~\ref{fig::superradiance}. Another example that falls under our framework is the fidelity of coherence distillation~\cite{Regula2018}.

\section{Conclusion}

In this article, we demonstrated that every nontrivial Hermitian observable $M$ corresponds to a coherence witness and to every coherence witness, there corresponds a coherence measure $\mathcal{C}^\mathcal{O}_{M}$, where the set of operations $\mathcal{O}$ may be either MIO or IO. In the case of MIO, we show that the measure is in fact always computable via a semidefinite program, leading to an infinitely large set of coherence measures. The measures also show that the task of optimizing $\langle M \rangle$ is the same as the task of maximizing the coherence of the input state, up to the application of some incoherent operation (Theorem~\ref{thm::measure}). They therefore have the operational interpretation of the usefulness of a given quantum state $\rho$ for the purpose or optimizing the observable $\langle M \rangle$. These coherence measures also satisfy a hierarchy $\mathcal{C}^{\mathrm{IO}}_M(\rho) \leq \mathcal{C}^{\mathrm{MIO}}_M(\rho) \leq \mathcal{N}_M \mathcal{C}_\mathcal{R}(\rho) \leq \mathcal{N}_M\mathcal{C}_{l_1}(\rho) $ (Theorem~\ref{thm::hierarchy}). This demonstrates that the robustness of coherence $\mathcal{C}_{\mathcal{R}}$ has an additional physical interpretation as the ultimate usefulness of a state $\rho$ for the purpose of optimizing \textit{any} obervable $M$. The $\l_1$ norm of coherence $\mathcal{C}_{l_1}$ is also interesting because it is expressible in a closed form formula, in comparison to $\mathcal{C}^\mathcal{O}_{M}$ and $\mathcal{C}_{\mathcal{R}}$ which both requires numerical optimization to compute.

A key conclusion of our results is that coherence witnesses and computable measures are in fact plentiful. This may in many cases allows coherence to be verified in the laboratory by simply inferring them from the existing measurement outcomes, without requiring additional specialized equipment. Moreover, the measurement outcomes of such observables are always, up to a constant displacement, a lower bound to a coherence measure $\mathcal{C}^\mathcal{O}_{M}$. Moreover, due to the hierarchy of coherence measures, they can alsko be used to find non-trivial lower bounds to the robustness of coherence and the $l_1$ without requiring full quantum state tomography. We hope that the techniques presented here will be useful to simplify the requirements for teh detection of nonclassical quantum effects in the laboratory, as well as allow new interesting coherence measures with novel physical interpretations to be discovered.

\acknowledgments This work was supported by the National Research Foundation of Korea (NRF) through a grant funded by the Korea government (MSIP) (Grant No. 2010-0018295) and by the Korea Institute of Science and Technology Institutional Program (Project No. 2E27800-18-P043). K.C. Tan was supported by Korea Research Fellowship Program through the National Research Foundation of Korea (NRF) funded by the Ministry of Science and ICT (Grant No. 2016H1D3A1938100). S. C. was supported by the Global PhD Fellowship Program through the NRF funded by the Ministry of Education (NRF-2016H1A2A1908381).

\end{document}